\documentclass[12pt,USenglish,3p]{elsarticle}
 
  \usepackage{amsmath,amsfonts}
  \usepackage{latexsym,amssymb}
\usepackage{array}
  \usepackage[normalem]{ulem}
  \usepackage{amsthm}
  \usepackage{tikz}
  \usepackage[T1]{fontenc}
  \usepackage[utf8]{inputenc}
  \usepackage{comment}
  \usepackage{xspace}
  \usepackage{enumerate}
  \usepackage{listings}
  \usepackage[hidelinks]{hyperref}
  \usepackage[noline,noend]{algorithm2e}
  \usetikzlibrary{decorations.pathreplacing,calc,snakes} 
\usepackage{microtype}
\usepackage[whileod]{chl-alg}
\usepackage{todonotes}

\newcommand{\String}{S}

\DeclareMathOperator{\Oh}{\mathcal{O}}

\theoremstyle{plain}
\newtheorem{theorem}{Theorem}

\newtheorem{lemma}[theorem]{Lemma}

\theoremstyle{definition}
\newtheorem{definition}[theorem]{Definition}
\newtheorem{example}{Example}

\theoremstyle{remark}

\begin{document}

\begin{frontmatter}

\title{Finding Maximal Closed Substrings}

\author{Golnaz Badkobeh} 
\ead{golnaz.badkobeh@city.ac.uk}
\address{City, University of London, London, UK}

\author{Alessandro De Luca} 
\ead{alessandro.deluca@unina.it}
\address{DIETI, Universit\`a degli Studi di Napoli Federico II, Italy}

\author{Gabriele Fici} 
\ead{gabriele.fici@unipa.it}
\address{Dipartimento di Matematica e Informatica, Universit\`a di Palermo, Italy}

\author{Simon J. Puglisi} 
\ead{simon.puglisi@helsinki.fi}
\address{
Helsinki Institute for Information Technology,\\
Department of Computer Science, University of Helsinki, Helsinki, Finland}

\sloppy  
  
\begin{abstract}
A string is closed if it has length 1 or has a nonempty border without internal occurrences. In this paper we introduce the definition of a \emph{maximal closed substring} (MCS), which is an occurrence of a closed substring that cannot be extended to the left nor to the right into a longer closed substring. MCSs with exponent at least $2$ are commonly called \emph{runs}; those with exponent smaller than $2$, instead, are particular cases of \emph{maximal gapped repeats}. We  provide an 
 algorithm that, given a string of length $n$ locates all MCSs the string contains in 
$\mathcal O(n\log n)$ time.
\end{abstract}


\end{frontmatter}


\section{Introduction}

A string is \emph{closed} (or \emph{periodic-like}~\cite{DBLP:journals/acta/CarpiL01}) if it has length $1$ or it has a border that does not have internal occurrences (i.e., it occurs only as a prefix and as a suffix). Otherwise, the string is \emph{open}. For example, the strings $a$, $abaab$ and $ababa$ are closed, while $ab$ and $ababaab$ are open. In particular, every string whose exponent -- the ratio between the length and the minimal period -- is at least $2$, is closed~\cite{BaFiLi15}. The distinction between open and closed strings was introduced by the third author in~\cite{Fici17} in the context of Sturmian words. 

In this paper, we consider occurrences of closed substrings in a string, with the property that the substring cannot be extended to the left nor to the right into another closed substring. We call these the \emph{maximal closed substrings} (MCS) of the string. For example, if $\String=abaababa$, then the set of pairs of starting and ending positions of the MCSs of $\String$ is
\[\{(1,1),(1,3),(1,6),(2,2),(3,4),(4,8),(5,5),(6,6),(7,7),(8,8)\}.\]
This notion encompasses that of a \emph{run} (maximal repetition) which is an MCS with exponent $2$ or larger. It has been conjectured by Kolpakov and Kucherov~\cite{DBLP:conf/focs/KolpakovK99} and then finally proved, after a long series of papers, by Bannai et al.~\cite{DBLP:journals/siamcomp/BannaiIINTT17}, that a string of length $n$ contains less than $n$ runs.

On the other hand, maximal closed substrings with exponent smaller than $2$ are particular cases of \emph{maximal gapped repeats}~\cite{DBLP:journals/jda/KolpakovPPK17}. An $\alpha$-gapped repeat ($\alpha \geq 1$) in a string $\String$ is a substring $uvu$ of $\String$ such that $|uv| \leq  \alpha |u|$. It is maximal if the two occurrences of $u$ in it cannot be extended simultaneously with the same letter to the right nor to the left. Gawrychowski et al.~\cite{DBLP:journals/mst/GawrychowskiIIK18} proved that there are strings that have $\Theta(\alpha n)$ maximal $\alpha$-gapped repeats.
In particular, as shown by Brodal et al.~\cite{BLPS99}, the string $(aab)^{n/3}$ contains $\Theta(n^2)$ maximal gapped repeats. However, not all such maximal gapped repeats are MCSs; for example the prefix $a\cdot aba\cdot a$ is a maximal gapped repeat, but it is not an MCS since it is extended to the right into the closed substring $aabaab$. In fact, the MCSs in this string are the whole string and all the occurrences of $aa$, $b$, and $aba$.

\medskip

In this paper, we consider how many MCSs can a string of length $n$ contain and  algorithms to find all MCSs in a given string.

We describe an algorithm that, given a string of length $n$, locates all MCSs the string contains in 
$\mathcal O(n\log n)$ time, thus implying 
a string of length $n$ contains $\mathcal O(n \log n)$ MCSs.
An algorithm for binary strings was first sketched in~\cite{DBLP:conf/spire/Badkobeh0FP22}. Here we generalize the algorithm to strings on larger alphabets and provide a more thorough description and analysis.

\section{Preliminaries}

Let $\String=\String[1..n]=\String[1]\String[2]\cdots \String[n]$ be a string of $n$ letters drawn from a fixed alphabet $\Sigma$. The length $n$ of a string $\String$ is denoted by $|\String|$. The \emph{empty string}, denoted by $\varepsilon$, has length $0$. A \emph{prefix} (resp.~a \emph{suffix}) of $\String$ is any string of the form $\String[1..i]$ (resp.~$\String[i..n]$) for some $1\leq i \leq n$. 
A \emph{substring} of $\String$ is any string of the form $\String[i..j]$ for some $1\leq i \leq j \leq n$. It is also commonly assumed that the empty string is a prefix, a suffix and a substring of any string. 

An integer $p\geq 1$ is a \emph{period} of $\String$ if $\String[i]=\String[j]$ whenever $i\equiv j \pmod p$. For example, the periods of $\String=aabaaba$ are $3$, $6$ and every $n\geq 7 = |\String|$. 

Given a string $\String$, we say that a string $\beta\neq \String$ is a \emph{border} of $\String$ if $\beta$ is both a prefix and a suffix of $\String$ (we exclude the case $\beta=\String$ but we do consider the case  $\beta=\varepsilon$). Note that if $\beta$ is a border of $\String$, then $|\String|-|\beta|$ is a period of $\String$; conversely, if $p\leq |\String|$
is a period of $\String$, then $\String$ has a border of length
$|\String|-p$. Furthermore,
if a string has two borders $\beta$ and $\beta'$, with $|\beta|<|\beta'|$, then $\beta$ is a border of $\beta'$. 

Central to our algorithm for MCSs are the concepts of suffix tree and binary suffix tree~\cite{W73}, which we now define.

\begin{definition}[Suffix tree]
The suffix tree $T(S)$ of the string $S$ is the compressed trie of all suffixes of $S$. Each leaf in $T(S)$ represents a suffix $S[i..n]$ of $S$ and is annotated with the index $i$. We refer to the set of indices stored at the leaves in the subtree rooted at node $v$ as the leaf-list of $v$ and let $LL(v)$ denote it. Each edge in $T(S)$ is labeled with a nonempty substring of $S$ such that the path from the root to the leaf annotated with index $i$ spells the suffix $S[i..n]$. We refer to the substring of $S$ spelled by the path from the root to the node $v$ as the path-label of $v$, denoted by $\hat{v}$.
\end{definition}

The suffix tree can be constructed in $O(n)$ time for strings on linearly sortable alphabets and $O(n\log n)$ time for general alphabets, provided symbols can be compared in constant time~\cite{F97}.
Because all internal nodes of a suffix tree $T(S)$ have outdegree between two and $|\Sigma|$, we can turn a suffix tree $T(S)$ into a {\em binary suffix tree} $T_B(S)$ by replacing every node $v$ in $T(S)$ with out-degree $d > 2$ by a binary tree with $d-1$ internal nodes and $d-2$ internal edges in which the $d$ leaves are the $d$
children of node $v$. We label each new internal edge with the empty string so that the $d-1$ nodes replacing node $v$ all have the same path-label as node $v$ has in $T(S)$ (see~Fig.~\ref{fig:trees}).
Since $T(S)$ has $n$ leaves, constructing the binary suffix tree $T_B(S)$ requires adding at
most $n-2$ new nodes. Since each new node can be added in constant time, the binary
suffix tree $T_B(S)$ can be constructed from $T(S)$ in $O(n)$ time~\cite{BLPS99}.

\begin{figure}
    \centering
    \begin{tikzpicture}
      \node (hid1) at (0,0) {}
        child {
          node [draw,circle] {$\phantom iv_{\phantom i}$}
          edge from parent [dashed]
            child [solid] {
              node [draw,circle] {$v_a$}
              edge from parent
                node[above left] {\footnotesize $a$}
            }
            child [solid] {
              node [draw,circle] {$v_b$}
              edge from parent
                node[left] {\footnotesize $b$}
            }
            child [solid] {
              node [draw,circle] {$v_c$}
              edge from parent
                node[left] {\footnotesize $c$}
            }
            child [solid] {
              node [draw,circle] {$v_d$}
              edge from parent
                node[right] {\footnotesize $d$}
            }
            child [solid] {
              node [draw,circle] {$v_e$}
              edge from parent
                node[above right] {\footnotesize $e$}
            }
        }
        child [missing] {
          node {}
        }
        child [missing] {
          node {}
        };
      \node (hid2) at (8,1) {}
        child {
          node [draw,circle] {$\phantom iv_{\phantom i}$}
          edge from parent [dashed]
            child [solid] {
              node [draw,circle] {$\phantom v$}
                child {
                  node [draw,circle] {$v_a$}
                  edge from parent
                    node [above left] {\footnotesize $a$}
                }
                child {
                  node [draw,circle] {$v_b$}
                  edge from parent
                    node [above right] {\footnotesize $b$}
                }
              edge from parent
                node [above left] {\footnotesize $\varepsilon$}
            }
            child [missing] {
              node {}
            }
            child [solid] {
              node [draw,circle] {$\phantom v$}
                child {
                  node [draw,circle] {$v_c$}
                  edge from parent
                    node [above left] {\footnotesize $c$}
                }
                child {
                  node [draw,circle] {$\phantom v$}
                    child {
                      node [draw,circle] {$v_d$}
                      edge from parent
                        node [above left] {\footnotesize $d$}
                    }
                    child {
                      node [draw,circle] {$v_e$}
                      edge from parent
                        node [above right] {\footnotesize $e$}
                }
                  edge from parent
                    node [above right] {\footnotesize $\varepsilon$}
                }
              edge from parent
                node [above right] {\footnotesize $\varepsilon$}
            }
        }
        child [missing] {
          node {}
        }
        child [missing] {
          node {}
        };
      \draw[->] (2.5,-1.5) -- (3.3,-1.5);
    \end{tikzpicture}
    \caption{Example of a step for turning $T(\String)$ into $T_B(\String)$.}
    \label{fig:trees}
\end{figure}
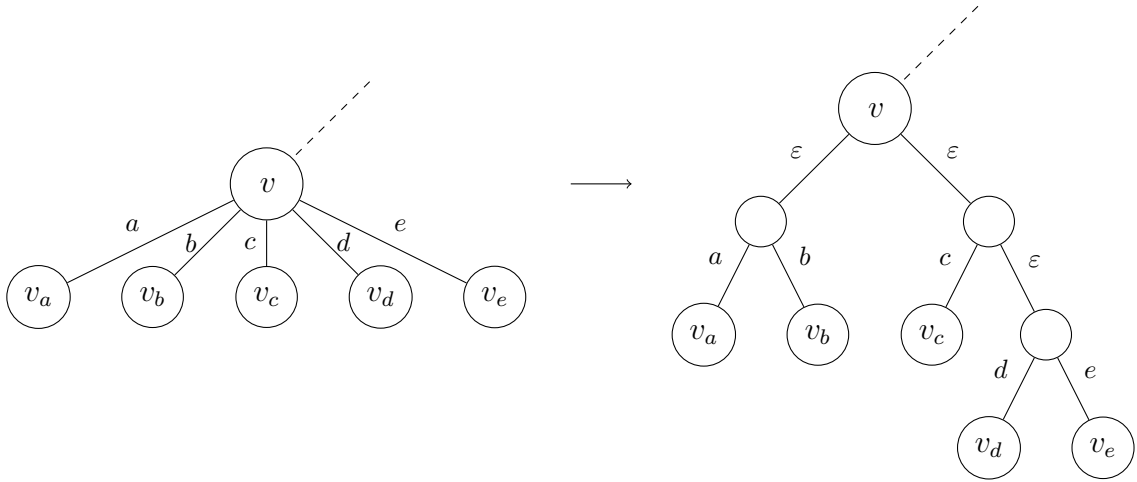

\section{An Algorithm for Locating All MCSs}

We now describe an algorithm for computing all the maximal closed substrings in a string $\String$ of length $n$. The algorithm uses the binary suffix tree $T_B(S)$ of the input string. The running time is $O(n\log n)$.
The inspiration for our approach is an algorithm for finding maximal pairs under gap constraints due to Brodal, Lyngs{\o}, Pedersen, and Stoye~\cite{BLPS99}.

At a high level, our algorithm for finding MCSs processes the binary suffix tree in a bottom-up traversal. At each node, the leaf lists of the (two, in a binary suffix tree) children are intersected. For each element in the leaf list of the smaller child, the successor in the leaf list of the larger child is found in a merge of the two trees. Note that because the element from the smaller child and its successor in the larger child come from different subtrees, they represent a pair of occurrences of the substring $\hat{v}$ that are right-maximal. To ensure left maximality, we must take care to only output pairs that have different preceding characters. We explain how to achieve this below.

Essential to our algorithm are properties of AVL trees that allow their efficient merging, and the so-called ``smaller-half trick'' applicable to binary trees. These proprieties are captured in the following lemmas.

\begin{lemma}[Brown and Tarjan~\cite{BT79}]
\label{lem:avlmerge}
Two AVL trees of size at most $n$ and $m$ can be merged in time $\mathcal O(\log \binom{n+m}{n})$.
\end{lemma}

\begin{lemma}[Brodal et al.~\cite{BLPS99}, Lemma 3.3]
\label{lem:smallerhalf}
Let $T$ be an arbitrary binary tree with $n$ leaves. The sum over all internal nodes $v$ in $T$ of terms that are $\mathcal O(\log \binom{n_1+n_2}{n_1})$, where $n_1$ and $n_2$ are the numbers of leaves in the subtrees rooted at the two children of $v$, is $\mathcal O(n \log n)$.
\end{lemma}

As stated above, our algorithm traverses the suffix tree bottom up.
We now describe the processing at a suffix tree node $v$ with path label $\hat{v}$. Recall that in the binary suffix tree, node $v$ from the original suffix tree is represented with a binary tree with $d$ leaves, where $d$ is the outdegree of the original suffix tree node. The processing described below results in the computation of all maximal closed substrings (MCSs) that have border $\hat{v}$.

At the start of the processing, for each child $v_a$ of $v$, we have an AVL tree the elements of which are the occurrence positions of the substring $\hat{v}_a$ in the string; in other words, the AVL tree for child $v_a$ contains $LL(v_a)$, the leaves in the subtree rooted at $v_a$ in the suffix tree. With each element $e$ in the AVL trees we associate two values: i) $c(e)$ is the position of a {\em candidate} right border so that $T[e..c(e)+|\hat{v}|]$ is potentially an MCS --- at the start of processing node $v$, $c(e) = 0$ for all $e$; and ii) $f(e)$ is a boolean flag indicating if $c(e)$ has ever been modified, initially false.

When processing an internal node of $v$'s binary tree, we are to merge two AVL trees $T_1$ and $T_2$, which contain, respectively, the leaf lists of the left and right child of the internal node. All elements in the AVL trees represent positions of occurrences of the substring $\hat{v}$ in the input string. In particular, observe that a pair of elements, one taken from $T_1$ and another from $T_2$, represents right maximal pair of occurrences of $\hat{v}$ because they are followed by different symbols.

We merge the two trees by merging the smaller one into the larger one so that the time spent merging is $O(n\log n)$ over the whole tree, according to Lemma~\ref{lem:avlmerge} and Lemma~\ref{lem:smallerhalf}. Without loss of generality, let $|T_1| \le |T_2|$. Before we actually merge the trees, we run a modified version of the merge algorithm as follows.

Let $c(e)$ denote the current candidate for position $e$. As stated above, initially $c(e) = 0$ for all $e$ and $c(e)$ is always stored at $e$'s AVL tree node. Additionally, we store a flag $f(e)$ at the node (initially false) indicating whether or not we have ever updated $c(e)$ in processing node $v$. The {\em first time} we ever set $c(e)$ when processing node $v$ (which we can tell from checking $f(e)$), we add $e$'s AVL tree node to a list $H$.

We process the elements $e_1 < e_2 < \ldots < e_{|T_1|}$ of $T_1$ in sorted order, which is easy because they are stored in an AVL tree. When the merge algorithm would insert one of the elements $e_i$ into $T_2$ we have the merge algorithm report the predecessor $p(e_i)$ and successor $s(e_i)$ of $e_i$ in $T_2$. 
This modification does not asymptotically increase the running time.
Now, consider $e_i$ and $s(e_i)$. If $e_{i+1} < s(e_i)$ then 
$T[e_i..s(e_i)+|\hat{v}|]$ is not a closed substring and we do not modify $c(e_i)$. Otherwise, if $s(e_i) < c(e_i)$ then we have determined $T[e_i..c(e_i)+|\hat{v}|]$ is not in fact closed, so we set $c(e_i) = 0$. However, $T[e_i..s(e_i)+|\hat{v}|]$ may be closed and if in addition we have $T[e_i-1] \ne T[s(e_i)-1]$ then the substrings are also left maximal, so we set candidate $c(e_i) = s(e_i)$\footnote{The processing for $p(e_i)$ and $e_i$ is symmetric.}.

The number of times $c(e_i)$ is updated is bound by the work done merging $T_1$ into $T_2$ (it is constant for each element of $T_1$). This also bounds the length of list $H$ containing the AVL tree nodes of elements that have had their candidates updated. When we are finished processing node $v$, we scan $H$ to output the MCSs having border $\hat{v}$. In particular, if the node for element $e$ is in $H$ and $c(e) \ne 0$ then we output $T[e..c(e)+|\hat{v}|]$ as an MCS. We then scan $H$ again to reset modified $c(e)$ values to 0, and reset flags to false at corresponding AVL tree nodes. We then empty $H$.

By Lemmas~\ref{lem:avlmerge} and~\ref{lem:smallerhalf}, the time to process the whole tree is bounded by $\mathcal O(n\log n)$. 
Thus, our algorithm is an implicit proof for an upper bound of $\mathcal O(n \log n)$ for the number of MCSs in a string of length $n$.

\section{Bounds on the Number of MCSs}

From the bound of our algorithm, we have the following

\begin{theorem}
A string of length $n$ contains $\Oh(n \log n)$ MCSs.
\end{theorem}

Runs of single letters are clearly MCSs. On the other hand, it is known that a string of length $n$ has less than $n$ runs, hence less than $n$ MCSs with exponent $2$ or larger. 

Take the string $(aab)^{n/3}$ considered in the Introduction. Its MCSs are the whole string and all the occurrences of $aa$, $b$, and $aba$. Therefore, this string contains $\Omega(n)$ MCSs. In particular, it contains $\Omega(n)$ MCSs that are neither runs nor occurrences of single letters, i.e., MCSs with exponent $\eta$ such that $1<\eta<2$.

We were not able to find a string with $\omega(n)$ MCSs. If such a string exists, then it must contain $\omega(n)$ MCSs that are neither runs nor occurrences of single letters.

One could be tempted to prove that at each position, only $O(\log n)$  MCSs start/end. However, this is not  true in general. For instance, a string can have $\Omega(n^{1/2})$ MCSs starting at the same position, as the following example shows:

 \begin{example}
     Let $\overline{a}=b$ and $\overline{b}=a$. Consider the sequence of strings defined by $\String_1=a$ and 
     \[\String_{k+1}=\String_{k} \overline{\String_{k}[k]}\String_{k}[1..k]\] 
     for $k\geq 1$. So, $\String_2=aba$,  $\String_3=abaaab$,  $\String_4=abaaabbaba$, etc.  The length of $\String_k$ is $k(k+1)/2$.
     Every closed prefix of $\String_k$ is an MCS. The closed prefixes of $\String_k$ are precisely the strings $\String_i$, for $i\leq k$.
 \end{example}

\section{Concluding Remarks}
This paper has studied maximal closed substrings (MCSs), a form of repetitions that generalize runs. Our main results have been to show that the maximum number of MCSs in a string is somewhere between $\Omega(n)$ and $O(n\log n)$. A tighter analysis of this maximum is the main open problem we leave, together with the problem of obtaining a constructive proof of the $O(n\log n)$ upper bound (as opposed to the non-constructive one implied by our algorithm). 

\subparagraph*{Acknowledgements}

Gabriele Fici is supported by MIUR project PRIN 2022 APML – 20229BCXNW. Simon J. Puglisi is supported by the Academy of Finland via grant 339070. 

\bibliographystyle{plain}

\end{document}